\documentclass[11pt, a4paper, reqno]{amsart}
\usepackage{amsthm, amsmath, amsfonts, amssymb, appendix, dsfont, latexsym}
\usepackage{amsfonts, mathrsfs}

\makeatletter
\usepackage{delarray, a4, color}
\usepackage{euscript}
\usepackage[latin1]{inputenc}
\usepackage{enumitem}
\usepackage{hyperref}
\hypersetup{
   colorlinks,
   menucolor=black,
   linkcolor=black,
   citecolor=black,
   urlcolor=blue
}

\theoremstyle{plain}
\newtheorem{theorem}{Theorem}[section]
\newtheorem*{theorem*}{Theorem}
\newtheorem{lemma}[theorem]{Lemma}

\newtheorem{proposition}[theorem]{Proposition}

\theoremstyle{definition}

\theoremstyle{remark}
\newtheorem{remark}[theorem]{Remark}
\newtheorem{remark*}[theorem]{Remark\textup{*}}

\numberwithin{equation}{section}

\DeclareMathAlphabet{\mathpzc}{OT1}{pzc}{m}{it}
\def\eps{\varepsilon}

\def\C {\mathbb{C}}

\def\N {\mathbb{N}}

\def\R {\mathbb{R}}

\def\Z {\mathbb{Z}}

\def\eD {\EuScript{D}}
\def\eN {\EuScript{N}}

\newcommand\1{{\ensuremath {\mathds 1} }}

\newcommand{\bA}{\mathbf{A}}

\newcommand{\bX}{\mathbf{X}}

\newcommand{\br}{\mathbf{r}}

\newcommand{\bx}{\mathbf{x}}
\newcommand{\by}{\mathbf{y}}

\newcommand{\cH}{\mathcal{H}}

\newcommand{\cS}{\mathcal{S}}

\newcommand{\hT}{\hat{T}}

\newcommand{\sD}{\textup{D}}

\newcommand{\sx}{\textup{x}}

\DeclareMathOperator{\infspec}{\mathrm{inf\, spec\,}}
\DeclareMathOperator{\curl}{\mathrm{curl}}

\DeclareMathOperator{\dom}{\mathrm{dom}}

\newcommand{\slot}{\cdot}

\newcommand{\domD}{\mathscr{D}}
\newcommand{\dist}{\mathrm{dist}}

\newcommand{\sym}{\mathrm{sym}}
\newcommand{\asym}{\mathrm{asym}}
\newcommand{\loc}{\mathrm{loc}}

\usepackage{mathtools}
\mathtoolsset{showonlyrefs=true}

\usepackage{placeins}

\title{Fermionic Behavior of Ideal Anyons}

\author[D. Lundholm]{Douglas LUNDHOLM}
\address{Department of Mathematics, KTH Royal Institute of Technology, SE-100 44 Stockholm, Sweden}
\email{dogge@math.kth.se}

\author[R. Seiringer]{Robert SEIRINGER}
\address{IST Austria, Am Campus 1, 3400 Klosterneuburg, Austria}
\email{robert.seiringer@ist.ac.at}


\begin{document}

\begin{abstract}
We prove upper and lower bounds on the ground-state energy of the
ideal two-dimensional anyon gas. Our bounds 
are extensive in the particle number, as for fermions, and 
linear in the statistics parameter $\alpha$.
The lower bounds extend to Lieb--Thirring inequalities for all anyons
except bosons.
\end{abstract}

\maketitle

\section{Introduction}\label{sec:intro}

	The behavior of quantum mechanical systems of particles depends
	sensitively on the geometry of the space in which the particles may move.
	In particular, 
	dimensionality plays a significant role, and it is a geometric fact that only 
	two fundamental types of identical particles naturally occur in three-dimensional space
	--- bosons and fermions, from whose basic statistical properties
	many collective quantum phenomena follow. More exotic possibilities of quantum statistics may be realized 
	by confining the particles' motion and thereby
	effectively lowering the dimensionality. 
	In two spatial dimensions, which we will be concerned with here, the richer topology allows for 
	a family of hypothetical quantum particles known as {\em anyons}.  
	
	Recall that the state of a quantum system of $N$ particles
	is described in terms of a Schr\"odinger wave function,
	$\Psi\colon (\R^2)^N \to \C$,
	whose amplitude $|\Psi(\sx)|^2$ 
	represents the probability density of finding the particles at positions
	$\sx = (\bx_1,\ldots,\bx_N)$, $\bx_j \in \R^2$. 
	If the particles are indistinguishable, one must impose that the
	density is symmetric under particle exchange, i.e., 
	\begin{equation}\label{eq:exchange-amp}
		|\Psi(\bx_1, \ldots, \bx_j, \ldots, \bx_k, \ldots, \bx_N)|^2
		= |\Psi(\bx_1, \ldots, \bx_k, \ldots, \bx_j, \ldots, \bx_N)|^2, 
		\quad j \neq k.
	\end{equation}
	This leaves the possibility for an exchange phase:
	\begin{equation}\label{eq:exchange-phase}
		\Psi(\bx_1, \ldots, \bx_j, \ldots, \bx_k, \ldots, \bx_N)
		= e^{i\alpha\pi} \Psi(\bx_1, \ldots, \bx_k, \ldots, \bx_j, \ldots, \bx_N), 
		\quad j \neq k.
	\end{equation}
	In the case of bosons ($\alpha=0$) or fermions ($\alpha=1$), one has $e^{i\alpha\pi} = \pm 1$, 
	so that a double exchange is trivial. However, by clarifying in topological terms 
	what exactly should be meant by 
	the exchange \eqref{eq:exchange-phase} (say a simple counterclockwise continuous exchange of two particles),
	it is possible to allow for \emph{any} phase 
	$e^{i\alpha\pi} \in U(1)$ or statistics parameter $\alpha \in \R$, 
	thereby defining a system of 
	\emph{any}ons\footnote{More precisely, these are \emph{abelian} anyons.
	Non-abelian anyons may be defined by replacing 
	complex phase factors by unitary matrices \cite{Froehlich-90,Nayak-etal-08}.}.
	Such possibilities have been known since the 1970s
	and have been studied extensively in the physics literature during the following decades, 
	with notable proposals 
	for concrete realizations and applications, such as 
	for quasi-particles in the fractional quantum Hall effect,
	rotating cold quantum gases,
	as well as for future prospects of quantum information storage and computation.
	We refer to 
	\cite{DatMurVat-03, Forte-92, Froehlich-90, IenLec-92, Khare-05, Lerda-92, LunRou-16, Myrheim-99, Nayak-etal-08, Ouvry-07, Stern-08, Wilczek-90} 
	for reviews.

	Mathematically, anyons can be realized by viewing $\Psi$ as a multi-valued 
	function or a section of a complex line bundle over 
	a nontrivial configuration manifold,
	an approach known in the literature as the anyon gauge picture \cite{DelFigTet-97}. 
	Alternatively, one can start with the usual quantum-mechanics setup, taking the familiar bosons or fermions as a  reference system,  
	and adding to these magnetic interactions of Aharonov--Bohm type \cite{LunSol-13a,LunSol-14,LunRou-15}. 
	Here we shall follow this latter approach, known as the magnetic gauge picture. 
		
	Many basic questions concerning the behavior of many-particle systems of anyons  have remained open
	since their discovery. 
	This is true even for ideal anyons, i.e., particles without any interactions in addition to the ones forced by statistics. 
	While non-interacting bosons and fermions admit a description solely
	in terms of the spectrum and eigenstates of the corresponding one-body 
	problem, allowing for the properties of the ideal quantum Bose and Fermi gases
	to be worked out easily,
	anyons with $0 < \alpha < 1$ do not admit such a simplification
	and must be treated within the full many-body context.
	Even their  ground-state properties are thus difficult to determine. In contrast, 
	recall that ideal bosons at zero temperature display complete Bose-Einstein condensation into a single 
	one-body state of lowest energy, while
	fermions are distributed over the $N$ lowest one-body states to satisfy the Pauli exclusion principle, leading in particular  
	to the extensivity of the fermionic ground-state energy.
	
	We show in this work that the ground-state energy of the ideal anyon gas 
	has a similar extensivity as the one for fermions, for all values of $\alpha$ except for zero (i.e., bosons).
	In fact, we shall derive upper and lower bounds that interpolate linearly 
	in $\alpha$ between bosons at $\alpha=0$ and fermions at $\alpha=1$.
	This improves on previous results which only applied to particular rational values of $\alpha$.
	Via well-known methods, our new bounds imply that also the celebrated Lieb--Thirring inequality holds for all anyons except for bosons. 

\section{Model and main results}\label{sec:intro-model}

	In the magnetic gauge formulation, 
	the kinetic energy operator for $N$ ideal (i.e., point-like) anyons in $\R^2$
	with statistics parameter $\alpha \in \R$
	is given by\footnote{We choose units such that $\hbar=1=2m$, with $m$ the particle mass.} 
	$$
		\hT_\alpha := \sum_{j=1}^N D_j^2,
	$$
	with the magnetically coupled momenta
	$$
		D_j := -i\nabla_{\bx_j} + \alpha\bA_j,
		\qquad
		\bA_j := \sum_{\substack{k=1 \\ k \neq j}}^N (\bx_j-\bx_k)^{-\perp}, 
	$$
	where
	$$
		\bx^{-\perp} := \frac{\bx^\perp}{|\bx|^2} = \frac{(-y,x)}{x^2+y^2} \qquad \text{for $\bx=(x,y)\in \R^2$}\, ,
	$$
	is the magnetic potential of an Aharonov-Bohm flux of magnitude $2\pi$ 
	at the origin, satisfying  $\curl \bx^{-\perp} = 2\pi\delta_0(\bx)$. 
	Since we demand that $\alpha=0$ represents bosons in accordance with~\eqref{eq:exchange-phase}, 
	we take the $N$-particle Hilbert space to be
	$\cH = L^2_\sym(\R^{2N})$, the permutation-symmetric square-integrable functions.
	The operators $\sD_\alpha = (D_j)_{j=1}^N$ and $\hT_\alpha$ 
	then act as unbounded operators on $\cH$ and,
	because of the singular nature of the vector potentials $\bA_j \notin L^2_\loc$,
	some care is needed to properly define their domains. 
	One can in fact show \cite[Theorem~5]{LunSol-14}
	that on $\R^{2N}$ the minimal and maximal realizations of $\sD_\alpha$
	coincide and hence induce a natural form domain 
	$\domD^N_\alpha = \dom(\sD_\alpha) \subset \cH$ 
	for the kinetic energy $\hT_\alpha$. 
	This choice 
	is then taken to model  ideal anyons. Indeed
	$\alpha=0$ yields free bosons, while $\alpha=1$ 
	corresponds to fermions,
	with their domains being the Sobolev spaces
	$\domD^N_0 = H^1_\sym$,  $\dom(\hT_0) = H^2_{\sym}$ 
	and $\domD^N_1 = U^{-1}H^1_\asym$, $\dom(\hT_1) = U^{-1}H^2_{\asym}$, respectively. 
	Here, the unitary map $U\colon L^2_{\sym/\asym} \to L^2_{\asym/\sym}$, 
	$$ 
		 (U\Psi)(\bx_1,\ldots,\bx_N) := \prod_{1 \le j<k \le N} \frac{z_j-z_k}{|z_j-z_k|} \Psi(\bx_1,\ldots,\bx_N),
		\quad z_j := x_j + iy_j,
	$$
	transforms bosons with attached unit magnetic flux
	into free fermions, and vice versa.
	
	In general, the gauge equivalence
	\begin{equation} \label{eq:gauge-equivalence}
		\sD_{\alpha+2n} = U^{-2n} \sD_{\alpha} U^{2n}, 
		\qquad
		\domD^N_{\alpha+2n} = U^{-2n} \domD^N_{\alpha}, 
		\qquad n \in \Z,
	\end{equation}
	with $U^{2n}\colon \cH \to \cH$,
	implies that the entire spectrum of $\hT_\alpha$ is $2\Z$-periodic in $\alpha$.
	It is also symmetric under the reflection $\alpha \mapsto -\alpha$, by complex conjugation $\Psi \mapsto \bar\Psi$.
	Note that these properties are all in line with the periodicity 
	of the exchange phase \eqref{eq:exchange-phase}. 
	In particular, it suffices to consider the case $0\leq \alpha\leq 1$ only, 
	which we will do from now on.
	
	When restricting to finite domains $\Omega \subset \R^2$ 
	the operator $\hT_\alpha$ and its spectrum  depends
	on the choice of boundary conditions.
	We may naturally define a Neumann realization 
	via the nonnegative quadratic form
	$$
		\langle\Psi,\hT_\alpha^{\Omega,\eN} \Psi\rangle 
		= \sum_{j=1}^N \int_{\Omega^N} |D_j \Psi|^2,
		\qquad \Psi \in \domD^N_\alpha,
	$$
	and a Dirichlet realization 
	$\hT_\alpha^{\Omega,\eD}$
	by considering the same form for $\Psi \in \domD^N_\alpha$ with compact support in $\Omega$.
	In particular, let us define the Neumann ground-state energy for $N$ 
	anyons on a domain $\Omega\subset \R^2$ as
	$$
		E_N^\eN(\alpha;\Omega) := \infspec \hT_\alpha^{\Omega,\eN}
		= 	\inf \left\{
			\sum_{j=1}^N \int_{\Omega^N} |D_j \Psi|^2
			\ : \  \Psi \in \domD^N_\alpha, \ \int_{\Omega^N} |\Psi|^2 = 1
			\right\}
	$$
	and likewise for the Dirichlet ground state energy 
	$E_N^\eD(\alpha;\Omega) = \infspec  \hT_\alpha^{\Omega,\eD}$. 

	For the special case of $\Omega$ equal to the unit square $Q_0 = [0,1]^2$, 
	we will drop $\Omega$ in the notation for simplicity, 
	and simply write $E_N^\eN(\alpha)$ and $E_N^\eD(\alpha)$, respectively. 
	Note that for a general square 
	$Q \subset \R^2$ we have
	\begin{equation}\label{eq:scaling}
		E^{\eN/\eD}_{N}(\alpha;Q) = |Q|^{-1} E^{\eN/\eD}_N(\alpha),
	\end{equation}
	due to the homogeneous scaling property of $\sD_\alpha$. 
	In particular, in the thermodynamic limit $N\to \infty$, $|Q|\to \infty$ 
	with the density $\rho=N/|Q|$ of the gas kept fixed, 
	the energy per particle is equal to $\rho$ times an $\alpha$-dependent 
	constant, given by $\lim_{N\to \infty} N^{-2} E^{\eN/\eD}_N(\alpha)$. 

	The case $\alpha=1$ corresponds to ideal fermions, where the ground state energy is obtained by simply adding up the $N$ lowest eigenvalues of the 
	one-body operator, i.e.,~the Laplacian $-\Delta_{Q_0}^{\eN/\eD}$. From the Weyl asymptotics one obtains 
	\begin{equation}\label{eq:Weyl-asymp}
		E^{\eN/\eD}_N(1) = 2\pi N^2 + o(N^2) \quad \text{as} \quad N \to \infty.
	\end{equation}
	On the other hand, for ideal bosons, i.e., $\alpha=0$, 
	$$
		E^{\eN}_N(0)=0 \qquad \text{and} \qquad E^{\eD}_N(0) = 2 \pi^2 N \,,
	$$
	which equals $N$ times the infimum of the spectrum of the Laplacian $-\Delta_{Q_0}^{\eN/\eD}$. 
	In the case $0<\alpha < 1$ of proper anyons there
	is no simplification to a one-body problem, however, and the system must
	be treated as a fully interacting many-body system. 
		
	Our main result is to show that for $0<\alpha<1$, 
	$E^{\eN/\eD}_N(\alpha) \sim N^2$, as in the fermionic case, 
	with a prefactor that is of order  $\alpha$ 
	both in the upper and lower bounds.
	In this sense, the ideal anyon gas behaves fermionic, for any $\alpha > 0$. 
	Since $E^{\eD}_N(\alpha) \geq E^\eN_N(\alpha)$, it is natural to derive an upper 
	bound on $E^{\eD}_N(\alpha)$ and a lower bound on $E^{\eN}_N(\alpha)$.

	Our main result is as follows:	

\begin{theorem}[\textbf{Bounds for the ideal anyon gas}]
	\label{thm:main}
	There exist constants $0 < C_1 \leq C_2<\infty$ 
	such that for any $0\leq \alpha \leq 1$,
	\begin{equation}\label{eq:main-bound}
		C_1\alpha N^2 \left( 1 - O(N^{-1}) \right) \leq E^\eN_N(\alpha) \leq E^\eD_N(\alpha) \leq C_2 \alpha N^2 + O(N)  \quad \text{as $N\to\infty$.} 
	\end{equation}
	Moreover, in the limit $\alpha \to 0$,
	\begin{equation}\label{eq:main-bound-asymptotic}
		\liminf_{N\to\infty} \frac{ E^\eN_N(\alpha)}{N^2} \ge \frac{\pi}{4} \alpha \bigl(1-O(\alpha^{1/3})\bigr)\,.
	\end{equation}
\end{theorem}

	These results should be compared with previous results in \cite{LarLun-16} and \cite{LunSol-13b}, respectively. 
	In  \cite{LunSol-13b} the upper bound
	$$ 
	E^\eD_N(\alpha)/N^2 \le 2\pi^2 + O(N^{-1/2}) \quad \text{independently of $\alpha$}
	$$
	was derived (the constant was not made explicit however). 
	In \cite[Theorem~1.5]{LarLun-16}, lower bounds were given 
	utilizing methods developed in 
	\cite{LunSol-13a,LunSol-13b,LunSol-14}
	to bound the anyon interaction in terms of an effective pair
	interaction, which is  of long range and has a coupling strength that depends
	on number-theoretic properties of $\alpha$.
	Namely, for rational $\alpha$ of the form of a reduced fraction 
	$\alpha = \mu/\nu$ with $\mu,\nu\in \N$, $\nu \geq 2$ and $\mu$ odd, one defines $\alpha_* := 1/\nu$,
	and $\alpha_* := 0$ otherwise. Note that $\alpha_* > 0$
	if and only if $\alpha$ is an odd-numerator rational. 
	The result of \cite[Theorem~1.5]{LarLun-16} is that 
	\begin{equation}\label{eq:main-bound-asymptotic-old}
		\liminf_{N\to\infty} \frac{ E^\eN_N(\alpha)}{N^2} \ge \left\{ \begin{array}{ll}
		\tilde C_1 \alpha_* & \text{for some constant $\tilde C_1>0$} \\ 
		{\pi}  \alpha_*  \bigl(1-O(\alpha_*^{1/3})\bigr)  &  \text{as $\alpha_*\to 0$.} \end{array} \right.
	\end{equation}
	While our lower bound \eqref{eq:main-bound-asymptotic} is weaker by a factor 4 for small $\alpha$ if $\alpha=\alpha_*$, it is valid for {\it all} $\alpha$, not just odd-numerator rationals. 

	Theorem~\ref{thm:main} answers a question raised in \cite{LunSol-13a,LunSol-13b}
	whether for $\alpha_* = 0$ (and $\alpha\neq 0$) the energy $E^{\eN/\eD}_N(\alpha)$  could be of lower 
	order in $N$ than the one for fermions or anyons with $\alpha_* > 0$. 
	It shows that the behavior of the ground-state energy is fermionic, for any $\alpha\neq 0$. 
	However, it still leaves open the possibility that the exact energy in the thermodynamic limit
	may be smaller around even-numerator rational $\alpha$
	than around $\alpha$ with relatively large $\alpha_*$, 
	i.e.\ odd-numerator rationals with small denominator. 
	In particular, it is not known whether it depends smoothly, 
	or even continuously, on $\alpha$.
	We refer to \cite{CorLunRou-proc-17,CorLunRou-16,Lundholm-16,Lundholm-17,LunSol-13b}
	for further discussion on 
	the $\alpha$-dependence of the ground-state energy.	
			
	The improved lower bounds in Theorem~\ref{thm:main} can be used to show the validity of a Lieb--Thirring inequality for anyons, extending the result derived in \cite{LunSol-13a}.
	Originally, Lieb and Thirring considered fermions in the context of
	stability of interacting Coulomb systems \cite{LieThi-75,LieThi-76} 
	(see also \cite{LieSei-09}), 
	and proved a uniform bound 
	for the kinetic energy of any fermionic many-body wave function $\Psi$
	in terms of an $L^p$-norm of its one-body density, defined as 
	\begin{equation}\label{def:rho}
		\varrho_\Psi(\bx) :=N  \int_{\R^{2(N-1)}} |\Psi(\bx, 
		\bx_{2},  \ldots, \bx_N)|^2 \prod_{k \geq 2}d\bx_k, 
	\end{equation}
	(in fact, $p=2$ in two dimensions) 
	thereby combining the uncertainty and Pauli exclusion principles 
	of quantum mechanics into a single powerful bound.
	In \cite[Theorem~1]{LunSol-13a}, an inequality of this type was proved to 
	hold for anyons in case $\alpha_* > 0$, with a quadratic dependence on $\alpha_*$, 
	and  was later improved in \cite[Theorem~1.6]{LarLun-16} 
	where a linear dependence in $\alpha_*$ was obtained. 
	Here we extend these results to all anyons except for bosons, i.e., any $0<\alpha\leq 1$. 
	
\begin{theorem}[\textbf{Lieb--Thirring inequality for ideal anyons}]
	\label{thm:LT} 
	There exists a constant $C>0$ such that
	for any $0\leq \alpha\leq 1$, $N \ge 1$ and $\Psi \in \domD_\alpha^N$
	\begin{equation} \label{eq:intro-LT-kinetic}
		\sum_{j=1}^N \int_{\R^{2N}} |D_j\Psi|^2
		\ge C \alpha  \int_{\R^2} \varrho_\Psi(\bx)^2 \, d\bx\,.
	\end{equation}
\end{theorem}
	
	One simple consequence of Theorem~\ref{thm:LT} concerns the ground-state energy of $\hat T_\alpha + \hat V$, where $\hat V(\bx_1,\ldots,\bx_N) := \sum_{j=1}^N V(\bx_j)$ for a one-body potential  $V\colon \R^2 \to \R$. One gets
	\begin{equation} \label{eq:intro-LT-potential}
		\infspec \left( \hat T_\alpha + \hat V\right) 
				\ge -\frac {1}{4C  \alpha} \int_{\R^2} V_-(\bx)^2 \, d\bx
	\end{equation}
	independently of $N$, 
	where  $V_{-} := \max\{- V, 0\}$ denotes the negative part of $V$. 
	Applying this, e.g., to $V(\bx)=|\bx|^2 - \mu$ and optimizing over $\mu>0$ gives the lower bound $\frac 43 N^{3/2} \sqrt{C\alpha/\pi}$ 
	on the ground state energy of the ideal anyon gas in a harmonic oscillator potential.

	The bound~\eqref{eq:intro-LT-potential} may for example be
	applied in a physically relevant 
	setting involving several species of charged 
	particles subject to Coulomb interactions
	and confined to a very thin two-dimensional layer.
	Taking one of the species of particles in the layer to be anyons, 
	as was previously considered in \cite[Theorem~21]{LunSol-14}, 
	our result proves that such a system is thermodynamically stable for any type 
	of anyon except for bosons.
	Our method of proof also clarifies that, at least in two dimensions, 
	stability is a consequence solely of the local two-particle repulsive properties 
	of any of the component species,
	in the sense that all that is required is 
	a strictly positive energy $E_2^\eN(\alpha)$,
	generalizing the Pauli exclusion principle.

\medskip\noindent\textbf{Acknowledgments.} 
D.\,L. would like to thank Simon Larson for discussions.
Financial support from the Swedish Research Council, 
grant no. {2013-4734} (D.\,L.), the European Research Council (ERC) under the European Union's Horizon 2020 research and innovation programme (grant agreement No 694227, R.\,S.), and by the Austrian Science Fund (FWF), project Nr. P 27533-N27 (R.\,S.), is gratefully acknowledged.

\section{Upper bounds}\label{sec:upper-bounds}

	A key tool for obtaining upper bounds is to use the fact that interactions between particles with wave functions supported on disjoint sets can be gauged away, as described in  \cite{LunSol-13b}. In fact, we have the following subadditivity property for the Dirichlet energy $E_N^{\eD}(\alpha;\Omega)$ on a general domain $\Omega\subset \R^2$. 
	
\begin{lemma}[Subadditivity]\label{lem:subadd}
If $\Omega_1$ and $\Omega_2$ are disjoint 
and simply connected subsets of $\R^2$, then
$$
E_{N_1+N_2}^{\eD}(\alpha;\Omega_1\cup\Omega_2) \leq E_{N_1}^\eD(\alpha;\Omega_1) + E_{N_2}^\eD(\alpha;\Omega_2)
$$
for any $0\leq \alpha\leq 1$ and $N_1,N_2\geq 1$. 
\end{lemma}

\begin{proof}
	Let $\Phi_1(\bx_1,\ldots,\bx_{N_1})$ be a function 
	in $\domD_\alpha^{N_1}$
	supported on $\Omega_1^{N_1}$, and similarly for $\Phi_2$ supported on $\Omega_2^{N_2}$. As a trial state for the $N_1+N_2$ particle problem, we can take 
		$$
		\Psi(\sx) = \cS\Biggl[
			\Phi_1(\bx_1,\ldots,\bx_{N_1}) \Phi_2 (\bx_{N_1+1},\ldots,\bx_{N_1+N_2} ) 
			\prod_{1\leq j \leq N_1 < k \leq N_1+N_2 } e^{-i\alpha\phi_{jk}}
			\Biggr]
	$$
	where
	\begin{equation}\label{def:phi}
		\phi_{jk} = \arg \frac{z_j-z_k}{|z_j-z_k|} ,
		\quad z_j := x_j + iy_j,
	\end{equation}
	and $\cS$ denotes symmetrization. 
	The phase factor $\phi_{jk}$ is a-priori only defined modulo $2\pi$, but can be chosen in a smooth way for $z_j\in \Omega_1$, $z_k\in \Omega_2$ due to our assumptions  on these domains.
	A simple calculation shows that 
	$$
	\sum_{j=1}^{N_1+N_2} \int_{(\Omega_1\cup\Omega_2)^{N_1+N_2} } |D_j \Psi|^2 = \sum_{j=1}^{N_1} \int_{\Omega_1^{N_1} } |D'_j \Phi_1|^2 + \sum_{j=1}^{N_2} \int_{\Omega_2^{N_2} } |D''_j \Phi_2|^2 
	$$
	where  $D_j' = -i\nabla_j + \alpha\sum_{1\leq k \leq N_1,\, k\neq j } (\bx_j-\bx_k)^{-\perp}$ and likewise for $D_j''$ 
	(involving only the particles in $\Omega_2$). The claimed bound readily follows.
	\end{proof}

The following lemma gives an upper bound on $E^\eD_N(\alpha)$ that is linear in $\alpha$ for small $\alpha$. It is restricted to small particle number, however. The bound follows from a calculation using a trial state similar to the one introduced by Dyson in \cite{Dyson-57} to obtain an upper bound on the ground state energy of the hard-sphere Bose gas.

\begin{lemma}[Upper bound {\`a} la Dyson]\label{lem:upper-dyson}
	If $8\pi \alpha N < 1$, then 
	\begin{equation}\label{eq:upperD}
        E_N^\eD(\alpha) \le 2 \pi^2 N
		+ \frac {9\pi} 2  N (N-1)  \alpha \frac {1 + \left(\frac 43\right)^3 20 \pi (N-2)\alpha  } { \left( 1 - 8 \pi \alpha N\right)^2}.
    	\end{equation}
    Furthermore, if $2\pi \alpha N <1$ then
	\begin{equation}\label{eq:upperN}
        E_N^\eN(\alpha) \le   2 \pi N (N-1) \alpha \frac {1 + \frac {20}{3} \pi (N-2)\alpha  } { \left( 1 - 2 \pi \alpha N\right)^2}.
    	\end{equation}
\end{lemma}

\begin{proof}
	We choose as a trial state a real-valued function $\Phi$, in which case 
	\begin{equation}
		\sum_{j=1}^N \int_{Q_0^N} |D_j \Phi|^2
		= \sum_{j=1}^N \int_{Q_0^N} \bigl( |\nabla_j\Phi|^2 + \alpha^2|\bA_j|^2 \Phi^2 \bigr), \label{eq:BG}
	\end{equation}
	which is the energy of an $N$-body Bose gas with two- and three-body interactions of the form 
	\begin{align}\nonumber
		\sum_{j=1}^N |\bA_j|^2 
		& = \sum_{j=1}^N \sum_{k \neq j} \sum_{l \neq j} (\bx_j-\bx_k)^{-\perp} \cdot (\bx_j-\bx_l)^{-\perp}
		\\ &= \sum_{j \neq k} |\bx_j-\bx_k|^{-2} 
		+ \sum_{j \neq k \neq l \neq j} (\bx_j-\bx_k)^{-\perp} \cdot (\bx_j-\bx_l)^{-\perp}. \label{eq:BGint}
	\end{align}
	It is well known that the minimum of the right side of \eqref{eq:BG} over {\em all} functions $\Phi$ is the same as the one over only bosonic $\Phi$ (see e.g. \cite[Corollary~3.1]{LieSei-09}), hence we may choose a $\Phi$ that is not permutation-symmetric. In particular, we can use a Dyson ansatz \cite{Dyson-57,LieSeiYng-00,LieYng-01}  of the form
	\begin{equation}\label{eq:vard}
	\Phi(\bx_1,\ldots,\bx_N) = \prod_{j=1}^N \varphi(\bx_j) f(\bx_j - \by_j(\bx_j;\bx_1,\ldots,\bx_{j-1})) 
	\end{equation}
	where we take $\varphi(\bx) = 2 \sin(\pi x) \sin(\pi y)$ to be the 
	$L^2$-normalized ground state of the Dirichlet Laplacian on $Q_0$, 
	$f$ is a nonnegative radial function bounded by $1$, 
	and $\by_j(\bx_j;\bx_1,\ldots,\bx_{j-1})$ denotes the nearest neighbor 
	of $\bx_j$ among the points $\{\bx_1,\ldots,\bx_{j-1}\}$. 
	A straightforward generalization of the calculation in 
	\cite{Dyson-57,LieSeiYng-00,LieYng-01} leads to the upper bound
	\begin{align*}
	 E_N^\eD(\alpha)& \leq  \|\Phi\|^{-2}  \sum_{j=1}^N \int_{Q_0^N} \bigl( |\nabla_j\Phi|^2 + \alpha^2|\bA_j|^2 \Phi^2 \bigr) \\ 
		& \leq 2 \pi^2 N 
			+  N(N-1)  \| \varphi\|_4^4  \frac {  \int_{\R^2} |\nabla f|^2 + \alpha^2 \int_B |f(\bx)|^2 |\bx|^{-2} } { \left( 1 - N \| \varphi\|_\infty^2 \int_{\R^2} (1 - f^2) \right)^2} \\ 
		& \quad + N(N-1)(N-2) \| \varphi\|_\infty^4  \frac { \frac 23  \left( \int_{\R^2} f |\nabla f| \right)^2 + \alpha^2 \left( \int_B |f(\bx)|^2 |\bx|^{-1} \right)^2  } { \left( 1 - N \| \varphi\|_\infty^2 \int_{\R^2} (1 - f^2) \right)^2}\,,
	 \end{align*} 
	 assuming that the term in parentheses in the denominators is strictly positive. Here $B$ denotes the ball of radius $\sqrt{2}$ centered at the origin. 
	 Note that $\|\varphi\|_4^4 = 9/4$ and $\|\varphi\|_\infty^2 = 4$. 
	 We shall choose 
	 $$
	 f(\bx) = \min\left\{ \left( |\bx| /  \sqrt{2} \right)^\alpha ,1 \right\}
	 $$
	 in which case
	 $$
	 \int_{\R^2} |\nabla f|^2 = \alpha^2 \int_B f(\bx)^2 |\bx|^{-2} =  \pi \alpha 
	 $$
	 as well as 
	 $$
	 \int_{\R^2} f|\nabla f|  = \alpha \int_B f(\bx)^2 |\bx|^{-1}  =  \frac{\sqrt{8}\pi\alpha}{1+2\alpha}    
	 \quad  \text{and} \quad  
	 \int_{\R^2} (1-f^2) = \frac {2\pi \alpha}{1+\alpha}.
	$$
	This leads to the claimed upper bound \eqref{eq:upperD}. 
	
	The same strategy can be used to obtain the upper bound \eqref{eq:upperN}
	on the Neumann energy $E^\eN_N(\alpha)$. 
	In this case, one simply chooses $\varphi=1$ in \eqref{eq:vard}. 
\end{proof}

A combination of Lemmas~\ref{lem:subadd} and \ref{lem:upper-dyson} leads to the following result, which immediately implies the upper bound claimed 
in~\eqref{eq:main-bound} in Theorem~\ref{thm:main}.

\begin{proposition}[Global upper bound]\label{lem:upper-linear}
	There exists a constant $C>0$ such that
	for any $0\leq \alpha \leq 1$ and any $N \ge 1$   we have
	$$
		E_N^\eD(\alpha) \le C \left( N + \alpha N^2\right)  \,.
	$$
\end{proposition}

\begin{proof}
	We shall divide the unit square $Q_0$ into disjoint smaller boxes and place a fixed number $n \ge 1$  particles in each box. 
	More precisely, we divide $Q_0$ into $M^2$ smaller boxes (squares) $\{ Q_q \}_{q=1}^{M^2}$ 
	of side length $M^{-1}$, with $M = \lceil (N/n)^{1/2} \rceil$.  
	We place $n$ particles into as many boxes as possible, and fewer than $n$ in the remaining ones, if necessary. 
	Denoting the number of particles in $Q_q$ by $n_q$, and using the 
	subadditivity in Lemma~\ref{lem:subadd} as well as the scaling property 
	\eqref{eq:scaling}, we obtain
	\begin{equation}\label{eq:use}
	E_N^\eD(\alpha) \leq M^2 \sum_{q=1}^{M^2} E^\eD_{n_q}(\alpha).
	\end{equation}
	
	We shall distinguish three cases. First, if $16\pi \alpha\geq 1$, we shall use \eqref{eq:use} for $n=1$. Since $E_1^\eD(\alpha) = 2 \pi^2$, we obtain
	$$
	E_N^\eD(\alpha) \leq 2 \pi^2 M^2 N \leq 2\pi^2 N \left( N^{1/2} + 1 \right)^2.
	$$
	In the opposite case $16\pi \alpha<1$, we shall choose $n$ such that $8\pi \alpha n < 1$, in which case we can apply the bound of Lemma~\ref{lem:upper-dyson} to $E^\eD_{n_q}(\alpha)$, and obtain
	$$
	E_N^\eD(\alpha) \leq M^2 \left( 2\pi^2 N  +  \frac{9 \pi}2 N n \alpha \frac {1 + \left( \frac 43\right)^3 20 \pi n\alpha  } { \left( 1 - 8 \pi \alpha n\right)^2}  \right) 
	$$
	using $n_q\leq n$ on each box. 
	Now if also $16 \pi \alpha N < 1$, we take $n=N$, i.e., $M=1$, and obtain 
	$$
	E_N^\eD(\alpha) \leq  2\pi^2 N  +  2 \pi N^2  \alpha \left( 9 + \frac {80}{3} \right).
	$$
	Finally, if $16\pi \alpha<1$ and $16 \pi \alpha N \geq  1$, we take 
	$n= \lfloor \frac 1{16\pi \alpha} \rfloor$
	so that $16\pi\alpha n \le 1$. 
	Then $M \leq \lceil (32 \pi \alpha N)^{1/2} \rceil \leq \frac 32  (32 \pi \alpha N)^{1/2} $, hence 
	$$
	E_N^\eD(\alpha) \leq 72 \pi \alpha N^2  \left( 2 \pi^2  +  \frac {9}{8} + \frac {10}{3}  \right).
	$$
	This completes the proof of the proposition.
\end{proof}

\section{Lower bounds}\label{sec:lower-bounds}

	As in \cite{LunSol-13a,LunSol-13b,LunSol-14,LarLun-16}, 
	the key ingredient in the strategy to obtain  lower bounds is to first prove 
	a lower bound for the local
	Neumann energy that is \emph{linear} in the particle number $N$.
	By splitting the original domain suitably one may then lift such a bound
	to one that is \emph{quadratic} in 
	$N$.
	This method and local bound, 
	referred to as a ``local exclusion principle'',
	goes back to the way Dyson and Lenard 
	incorporated the Pauli exclusion principle for fermions in their original 
	proof of stability of matter \cite{DysLen-67},
	and was further developed in \cite{LunSol-13a,LunSol-14,LunPorSol-15,LunNamPor-16} 
	for interacting bosonic gases 
	and in \cite{FraSei-12} for a model of fermions with point interactions.

\subsection{Preliminaries}\label{sec:prev-bounds}

	We start by recalling some of the previously obtained lower bounds 
	which shall also  turn out to be useful in deriving the new bounds.
	The simplest one is the usual diamagnetic inequality which is also
	valid for anyons \cite[Lemma~4]{LunSol-14}
	and tells us that their kinetic energy is always at least as big as the one of bosons:
	
\begin{lemma}[Diamagnetic inequality]\label{lem:diamagnetic}
	For any $0 \le \alpha \le 1$, $N \ge 1$, $\Omega\subset \R^2$ and $\Psi \in \domD^N_\alpha$ 
	one has the inequality
	$$
		\sum_{j=1}^N \int_{\Omega^N} |D_j\Psi|^2 
		\ge \sum_{j=1}^N \int_{\Omega^N} \bigl|\nabla_j|\Psi|\bigr|^2.
	$$
\end{lemma}

	Next we consider a certain analog of Lemma~\ref{lem:subadd} for the Neumann energy, where subadditivity becomes superadditivity.
	
\begin{lemma}[Superadditivity]\label{lem:superadd}
	For $K\geq 2$, let $\{\Omega_q\}_{q=1}^K$ be a collection of disjoint, 
	simply connected
	subsets of $\R^2$. For $\vec n \in \N_0^K$ with $\sum_{q} n_q = N$, 
	let $\1_{\vec n}$ denote the characteristic function of the subset of 
	$\R^{2N}$ where exactly $n_q$ of the points $\{\bx_1,\ldots,\bx_N\}$ 
	are in $\Omega_q$, for all $1\leq q\leq K$. Let 
	\begin{equation}\label{def:W}
		W(\bx_1,\dots,\bx_N) := \sum_{\vec n} \sum_{q=1}^K E^\eN_{n_q}(\alpha; \Omega_q) \1_{\vec n}(\bx_1,\ldots,\bx_N).
	\end{equation}
	With $\Omega:= \cup_q \Omega_q$, we have 
	\begin{equation}\label{eq:bound-splitting-potential}
		\sum_{j=1}^N \int_{\Omega^N} |D_j\Psi|^2 \ge \int_{\Omega^N} W|\Psi|^2
	\end{equation}
	for any $\Psi \in \domD^N_\alpha$. 
	In particular, 
	\begin{equation}\label{eq:superadd}
	E^\eN_N(\alpha;\Omega) \geq \min_{\vec n} \sum_{q=1}^K E^\eN_{n_q}(\alpha; \Omega_q).
	\end{equation}
\end{lemma}

\begin{proof}
	We start by noting that if $\bx_j \in \Omega$ for all $1\leq j\leq N$, then 
	$1 =  \sum_{\vec n} \1_{\vec n} (\bx_1,\dots,\bx_N)$. 
 	Moreover, for any given $\vec n$, we can further divide the support of $\1_{\vec n}$ into sets corresponding to a labeling of what particles are in what subset. Specifically,  for any $\Psi \in \domD^N_\alpha$
	$$
		\sum_{j=1}^N \int_{\Omega^N} |D_j\Psi|^2 \,d\sx = \sum_{\{A_k\}} \sum_{q=1}^K 
			\int_{(\Omega\setminus \Omega_q )^{N - |A_q|}}
			\sum_{j \in A_q} \int_{\Omega_q^{|A_q|}} |D_j \Psi|^2 
			\,d\sx_{A_q} \,d\sx_{A_q^c}\,,
	$$
	where the sum runs over all partitions of the particles into the sets $\Omega_q$,  
	i.e., over collections of disjoint subsets $A_k \subseteq \{1,2,\ldots,N\}$ 
	such that $|A_1| + \ldots + |A_K| = N$. We have introduced the notation $d\sx_A = \prod_{j\in A} d\bx_j$. 
	Note that, for given $q$, all the particles with labels in $A_q$ are located in $\Omega_q$, while the others 
	are located in $\Omega\setminus \Omega_q$. The interaction of particles inside and outside $\Omega_q$ can then be gauged away, as in the proof of Lemma~\ref{lem:subadd}, 
	explicitly by writing
	$\tilde\Psi = \prod_{j \in A_q, k \in A_q^c} e^{i\alpha\phi_{jk}} \Psi$, with $\phi_{jk}$ defined in \eqref{def:phi}:
	\begin{equation}\label{eq:Neumann-bracketing-apriori}
		\sum_{j\in A_q} \int_{\Omega_q^{|A_q|}} |D_j \Psi|^2 \,d\sx_{A_q} 
		= \sum_{j\in A_q} \int_{\Omega_q^{|A_q|}} |D_j' \tilde\Psi|^2 \,d\sx_{A_q} 
		\ge E^\eN_{|A_q|}(\alpha;\Omega_q) \int_{\Omega_q ^{|A_q|}} |\tilde\Psi|^2 \,d\sx_{A_q},
	\end{equation}
	where $D_j' = -i\nabla_j + \alpha\sum_{k \in A_q, \, k \neq j} (\bx_j-\bx_k)^{-\perp}$.
	Since $|\tilde\Psi|=|\Psi|$ we thus arrive at the desired lower bound \eqref{eq:bound-splitting-potential}.
\end{proof}

With the aid of the previous two lemmas, we can obtain the following bound, which is an adaptation of 
 \cite[Proposition~2]{Lundholm-13}. 
 
\begin{lemma}[A priori bounds in terms of $E_2^\eN(\alpha)$]\label{lem:apriori}
	For any $0 \le \alpha \le 1$ and $N \ge 3$ we have
	\begin{equation}\label{eq:apriori-limit}
	E_N^\eN(\alpha) \geq \frac{ \pi^2   \binom{N}{2} \left(\frac{3}{4}\right)^{N-2} E_2^\eN(\alpha)}{ \left( \pi + 4 \sqrt{E_2^\eN(\alpha)} \right)^2 + E_2^\eN(\alpha)  	\binom{N}{2} \left( \frac 3 4 \right)^{N-2} }
	\end{equation}
\end{lemma}

\begin{proof}
	Let us split $Q_0$ into four equally large squares,
	$Q_0 = Q_1 \sqcup Q_2 \sqcup Q_3 \sqcup Q_4$. Lemma~\ref{lem:superadd} implies that 
	\begin{equation}\label{eq:bound-splitting-potential2}
		\sum_{j=1}^N \int_{Q_0^N} |D_j\Psi|^2 \ge \int_{Q_0^N} W|\Psi|^2
	\end{equation}
	with $W$ defined in \eqref{def:W} (with $K=4$ and $\Omega_q = Q_q$ for $1\leq q\leq 4$). 
		If we keep only the terms in \eqref{def:W} where $n_q=2$, we obtain the lower bound 
	\begin{equation}\label{eq:splitting-potential-2}
		W(\bx_1,\ldots,\bx_N) \ge W_2(\bx_1,\ldots,\bx_N) := 4E_2^\eN(\alpha) \sum_{\vec n } \sum_{q=1}^4 (n_q=2) 
			\,\1_{\vec n}(\bx_1,\ldots,\bx_N),
	\end{equation}
	where we have introduced the convenient notation
	$(P)=1$ if the statement $P$ is true and $(P)=0$ otherwise, 
	and used the scaling property 
	$E_2^\eN(\alpha;Q_q) = 4 E_2^\eN(\alpha)$ for $1\leq q\leq 4$. 
	The average value of $W_2$ can be 
	computed to be
	$$
		 \int_{Q_0^N} W_2 
		=  E_2^\eN(\alpha)  \binom{N}{2}  \left(\frac 34\right)^{N-2}
	$$
	by counting the probability that exactly two particles are in a given square. 
	
	In order to estimate the expectation value of the potential $W_2$ in a ground state $\Psi$,   
	we borrow a bit of  kinetic energy and use the diamagnetic inequality of Lemma~\ref{lem:diamagnetic}. 
	That is, for arbitrary $\kappa \in [0,1]$ we write
	$$
		\hT_\alpha^{Q_0,\eN} = \kappa\hT_\alpha^{Q_0,\eN} + (1-\kappa)\hT^{Q_0,\eN}_\alpha 
		\ge \kappa\hT_\alpha^{Q_0,\eN}+ (1-\kappa)W_2\, .
	$$
	The diamagnetic inequality then  implies that 
	$$
	E_N^\eN(\alpha)  = \infspec 	\hT_\alpha^{Q_0,\eN} \geq \infspec \left[  -\kappa \Delta_{Q_0^N}^\eN + (1-\kappa)W_2 \right],
	$$
	with $\Delta_{Q_0^N}^\eN$ denoting the Neumann Laplacian on $Q_0^N$. 
	Consider the projection $P_0 := u_0\langle u_0,\slot\rangle$
	onto its normalized ground state $u_0 \equiv 1$,
	and the orthogonal complement $P_0^\perp = \1 - P_0$,
	for which we have
	$$
		-\Delta_{Q_0^N}^\eN \ge \pi^2 P_0^\perp\,.
	$$
	Since $W_2\geq 0$, the Cauchy-Schwarz inequality implies that
	$$
		W_2 = (P_0 + P_0^\perp)W_2(P_0 + P_0^\perp) 
		\ge (1-\eps)P_0 W_2 P_0 + (1-\eps^{-1})P_0^\perp W_2 P_0^\perp,
	$$
	for arbitrary $\eps\in (0,1)$. We have
	$$
		P_0 W_2 P_0  = P_0 \int_{Q_0^N} W_2 \,.
	$$
	By using also the simple bound
	$$
		P_0^\perp W_2 P_0^\perp  \le P_0^\perp \|W_2\|_\infty \le P_0^\perp \,2^4 E_2^\eN(\alpha),
	$$
	we obtain
	\begin{align*}
		 -\kappa \Delta_{Q_0^N}^\eN + (1-\kappa)W_2  &  \ge \left( \kappa\pi^2 - (1-\kappa)(\eps^{-1}-1)2^4 E_2^\eN(\alpha) \right) P_0^\perp 
			\\ & \quad + (1-\kappa)(1-\eps) \binom{N}{2} \left(\frac{3}{4}\right)^{N-2} E_2^\eN(\alpha)\, P_0.
	\end{align*}
	
	The optimal choice of $\kappa$ is to make the prefactors in front of the two projections on the right side equal, i.e., 
	$$
	\kappa =  \frac{ 2^4 (\eps^{-1}-1) E_2^\eN(\alpha)\left[ 1 + \eps \binom{N}{2} \frac{ 3^{N-2}}{4^N}\right]} { \pi^2 + 2^4 (\eps^{-1}-1) E_2^\eN(\alpha) \left[ 1 + \eps \binom{N}{2} \frac{ 3^{N-2}}{4^N}\right] }.
	$$
	This choice leads to the bound
	$$
	E_N^\eN(\alpha) \geq \frac{ \pi^2  (1-\eps) \binom{N}{2} \left(\frac{3}{4}\right)^{N-2} E_2^\eN(\alpha)}{ \pi^2 + 2^4 (\eps^{-1}-1) E_2^\eN(\alpha) \left[ 1 + \eps \binom{N}{2} \frac{ 3^{N-2}}{4^N}\right] }.
	$$
	Optimizing over $0<\eps<1$ then yields the claimed bound.
\end{proof}

\begin{remark}
Lemma~\ref{lem:apriori} implies, in particular, that $E_N^\eN(\alpha)$ is 
bounded below by a strictly positive, $N$-dependent constant times 
$E_2^\eN(\alpha)$. 
In fact, by localizing the two particles in different halves of the unit square $Q_0$ 
(following the proof of Lemma~\ref{lem:subadd}),
one readily checks that $E_2^\eN(\alpha) \leq 2 \pi^2$ independently of $\alpha$. 
Using this in the denominator in \eqref{eq:apriori-limit} leads to the simpler (but worse) bound
\begin{equation}\label{eq:apriori-uniform}
	E_N^\eN(\alpha) \geq \frac{   \binom{N}{2} \left(\frac{3}{4}\right)^{N-2} }{ \left( 1+ 4 \sqrt{2} \right)^2 + 2   	\binom{N}{2} \left( \frac 3 4 \right)^{N-2} } E_2^\eN(\alpha).
\end{equation}
Note that while this gives a non-zero bound for all $N\geq 2$, the constant appearing on the right side is exponentially small as $N\to \infty$. 
Moreover, from \eqref{eq:upperN} we deduce that $E_2^\eN(\alpha)\leq 4\pi \alpha (1 + O(\alpha))$ for small $\alpha$, hence \eqref{eq:apriori-limit} implies that
\begin{equation}\label{eq:apriori-limit-small}
	E_N^\eN(\alpha) \geq  \binom{N}{2} \left(\frac{3}{4}\right)^{N-2} E_2^\eN(\alpha) \bigl( 1 - O(\sqrt{\alpha}) \bigr)
	\end{equation}
	for small $\alpha$. 
\end{remark}

As a final step in this subsection, we shall give a lower bound on $E_2^\eN(\alpha)$. 
The following bound is actually contained in  \cite[Lemma~5.3]{LarLun-16}. 

\begin{lemma}[Lower bound on $E^\eN_2(\alpha)$]\label{lem:n2} 
	For  $\nu>0$ let $j_\nu'$ denote the first positive zero of the derivative 
	of the Bessel function $J_\nu$, 
	satisfying
	\begin{equation}\label{eq:jprime-bounds}
		\sqrt{2\nu} \leq j_\nu' \leq \sqrt{2\nu(1+\nu)},
	\end{equation}
	and $j_0' := 0$ for continuity. There exists a function $f\colon [0,(j_1')^2] \to \R_+$ 
		satisfying 
	\begin{equation}\label{eq:f-props}
		t/6 \le f(t) \le 2\pi t 
		\qquad \text{and} \qquad
		f(t) = 2\pi t \bigl(1 - O(t^{1/3})\bigr) \quad \text{as $t\to 0$,}
	\end{equation}
	such that 
	$$
		E^\eN_2(\alpha) \ge f\bigl( (j'_{\alpha})^2 \bigr) 
	$$
	holds for any $0\leq \alpha\leq 1$.
\end{lemma}	
	
In fact, the function $f$ in Lemma~\ref{lem:n2} is defined as 	
	$$
		f(t) := \frac{1}{2} \sup_{\kappa \in (0, 1)} 
		\inf_{\int_{Q_0^2}|\psi|^2=1} \int_{Q_0^2}
		\biggl(\kappa\bigl(|{\nabla_1|\psi|}\bigr|^2 
			+ \bigl|{\nabla_2|\psi|}\bigr|^2\bigr) 
		+ (1-\kappa) t \frac{\1_{B_{\delta(\bX)}}(\br)}{\delta(\bX)^2}|\psi|^2\biggr) 
		\,d\bx_1 d\bx_2,
	$$
	where $B_r$ denotes the ball of radius $r$ centered at the origin, and 
	$$
		\br = (\bx_1-\bx_2)/2, \qquad \bX = (\bx_1+\bx_2)/2, 
		\qquad \delta(\bx) := \dist(\bx,\partial Q_0)\,.
	$$

Note that in combination with the upper bound \eqref{eq:upperN}
of Lemma~\ref{lem:upper-dyson}, Lemma~\ref{lem:n2} determines the two-particle energy for small $\alpha$:

\begin{proposition}\label{prop:E_2-asymptotics}
	For the 2-anyon Neumann energy 
	\begin{equation}\label{eq:E_2-asymptotics}
		E^\eN_2(\alpha) = 4\pi\alpha \bigl(1 + O(\alpha^{1/3})\bigr)
		\qquad \text{as} \ \ \alpha \to 0.
	\end{equation}
\end{proposition}

\begin{remark}
	The bound in \cite[Lemma~5.3]{LarLun-16} is actually more general than what 
	is stated here. It 
	gives a lower bound, for any $N\geq 2$, in 
	terms of the `fractionality' of $\alpha$ \cite[Proposition~5]{LunSol-13a} 
	defined as
	\begin{equation}\label{eq:alpha_N}
		\alpha_N 
		:= \min\limits_{p \in \{0, 1, \ldots, N-2\}} \min\limits_{q \in \Z} |(2p+1)\alpha - 2q|, 
		\qquad
		\alpha_* = \inf_{N \ge 2} \alpha_N = \lim_{N \to \infty} \alpha_{N}.
	\end{equation}
	Note that $\alpha_2 = \alpha$ for $0\leq \alpha\leq 1$. 
	One has, in fact, for any $\alpha \in \R$ and $N \ge 1$ the bound
	$$
		E^\eN_N(\alpha) \ge f\bigl( (j'_{\alpha_N})^2 \bigr) (N-1)_+ \,.
	$$
	For $\alpha_*>0$, the right side grows linearly in $N$.
\end{remark}

\subsection{New bounds}\label{sec:new-bounds}

	Our improved lower bounds are due to the following lemma,
	which utilizes the scale invariance of the problem:

\begin{lemma}[$N$-linear bound in terms of few-particle energies]
	\label{lem:linear}
	For any $0 \le \alpha \le 1$ 
	and $N \ge 2$ we have
	\begin{equation}\label{eq:split}
		E_N^\eN(\alpha) \ge \frac{N}{4} \min\bigl\{ E_2^\eN(\alpha), E_3^\eN(\alpha), E_4^\eN(\alpha) \bigr\}.
	\end{equation}
\end{lemma}

\begin{proof}
	Without loss of generality we can assume $N\geq 5$, since for 
	$N\in \{2,3,4\}$ the bound \eqref{eq:split} trivially holds. 
	We may also assume $\alpha > 0$, so that $E^\eN_N(\alpha) > 0$ 
	for all $N \ge 2$ by Lemmas~\ref{lem:apriori} and \ref{lem:n2}.
	Let us proceed similarly as in the proof of Lemma~\ref{lem:apriori} and 
	split $Q_0 = Q_1 \sqcup Q_2 \sqcup Q_3 \sqcup Q_4$ 
	into four equally large squares. 
	The bound~\eqref{eq:superadd} together with the scaling property \eqref{eq:scaling} implies
\begin{equation}\label{eq:superadd2}
	E^\eN_N(\alpha) \geq 4 \min_{\vec n}  \sum_{q=1}^4 E^\eN_{n_j}(\alpha). 
	\end{equation}
	For any  partition $\vec n$
	of the $N$ particles into the four squares
	there must be at least one square with at least $N/4$ particles. Dropping the other terms, we thus obtain the recursive bound 
	\begin{equation}\label{eq:E_n-recursion}
		E^\eN_N(\alpha) \ge 4\min_{k=\lceil N/4\rceil, \ldots, N} E^\eN_k(\alpha).
	\end{equation}
	
	Let us define, for $k \ge 0$,
	$$
		e_k := \min_{n = 4^k+1,4^k+2, \ldots, 4^{k+1}} E^\eN_n(\alpha),
	$$
	and observe that by \eqref{eq:E_n-recursion}  
	$$
		e_k \ge 4 \min_{n = 4^k+1, \ldots, 4^{k+1}} 
			\,\min_{p= \lceil n/4\rceil, \ldots, n} E^\eN_{p}(\alpha)
			\ge 4 \min_{n= 4^{k-1} + 1, \ldots, 4^{k+1}} E^\eN_{n}(\alpha)
			= 4 \min \{ e_{k-1}, e_k \}
	$$
	for any $k \ge 1$. 
	Then, since $e_k > 0$ for all $k$, 
	we have
	\begin{equation}\label{eq:e_k-recursion}
		e_k \ge 4 e_{k-1} \ge \ldots \ge 4^k e_0.
	\end{equation}
	Finally, writing any $N \ge 5$ uniquely as $N = 4^k + l$ with $1\leq l \leq 3 \cdot 4^k$, 
	we have $k \ge 1$, $N \le 4^{k+1}$, and
	$$
		E^\eN_N(\alpha) \ge e_k \ge 4^k e_0 \ge \frac{N}{4} e_0,
	$$
	with $e_0 = \min \{E^\eN_2(\alpha),E^\eN_3(\alpha),E^\eN_4(\alpha)\}$. This proves the statement of the lemma.
\end{proof}

The previous lemma gives a lower bound on $E^\eN_N(\alpha)$ that is linear in $N$, at least for $N\geq 2$, for all $\alpha>0$. The following bound 
	(which also appeared in slightly different formulations in the earlier works; see 
	\cite{FraSei-12,LarLun-16,Lundholm-13}) 
	lifts any linear growth in the particle number $N$ to a quadratic one.

\begin{lemma}[Quadratic bounds]\label{lem:quadratic}
	If there exists an integer $k \ge 1$ and a function $c(\alpha) \ge 0$ such that
	$E_N^\eN(\alpha) \ge c(\alpha) (N-k)_+$ for all $N \ge 1$,
	then in fact
	$$
		E_N^\eN(\alpha) \ge c(\alpha) \frac{N^2}{4k} \bigl(1-O(k / N)\bigr)
		\quad \text{as} \quad N \to \infty.
	$$
\end{lemma}

\begin{proof}
	Given an integer $K\geq 1$, we split $Q_0$ into $K^2$ disjoint and equally large squares $\{Q_q\}_{q=1}^{K^2}$, 
	and associate to any $L^2$-normalized symmetric wave function $\Psi$
	the probabilities
	$$
		p_n(q) := \binom {N}{n} 
			\int_{(Q_q^c)^{N-n}\times  Q_q^n} |\Psi|^2 
	$$
	of finding exactly $n$ particles on a square $Q_q$. 
	Lemma~\ref{lem:superadd} implies that 
		\begin{align*}
		\sum_{j=1}^N \int_{Q_0^N} |D_j\Psi|^2 
		&\ge \sum_{q=1}^{K^2} \sum_{n=0}^N E^\eN_{n}(\alpha;Q_q) \,p_n(q)
		= K^2 \sum_{q=1}^{K^2} \sum_{n=0}^N E^\eN_n(\alpha) \,p_n(q) \\
		&\ge K^2 \sum_{q=1}^{K^2} \sum_{n=0}^N  c(\alpha) (n-k)_+ \,p_n(q) 
		= c(\alpha) K^4 \sum_{n=0}^N (n-k)_+ \gamma_n,
	\end{align*}
	where the average distribution of particle numbers
	$\gamma_n := K^{-2} \sum_{q=1}^{K^2} p_n(q)$
	satisfies
	$$
		\sum_{n=0}^N \gamma_n = 1
		\qquad \text{and} \qquad
		\sum_{n=0}^N n \gamma_n = N/K^2 =: \rho_Q,
	$$
	the expected number of particles on any square.
	Hence, by convexity of $x \mapsto (x-k)_+$,
	\begin{align*}
		\sum_{j=1}^N \int_{Q_0^N} |D_j\Psi|^2
		&\ge c(\alpha) K^4 \left( \sum_{n=0}^N n \gamma_n -k \right)_+
		= c(\alpha) N^2 \rho_Q^{-2} (\rho_Q - k)_+.
	\end{align*}
	In order to maximize the right side, the optimal choice of $K$ would be such as to make $\rho_Q = 2k$, in which case the desired 
	bound would be obtained exactly. However, we have to take into account the constraint that 
	$\rho_Q = N/K^2$ with $K \in \N$.
	Thus, taking $K := \lceil\sqrt{N/(2k)}\rceil$
	we obtain
	$$
		\frac{2k}{(1+ \sqrt{2k/N})^2} \le \rho_Q \le 2k 
	$$
	and
	$$
		E^\eN_N(\alpha) \ge c(\alpha) \frac{ N^2} {4k} 
			\left( 2(1+\sqrt{2k/N})^{2} - (1+\sqrt{2k/N})^{4} \right)_+,
	$$
	which proves the lemma.
\end{proof}

The proof of the lower bounds of Theorem~\ref{thm:main} now follows in a straightforward manner. 
	For any $\alpha \in \R$ and $N \ge 2$ we have by Lemma~\ref{lem:linear}
	\begin{equation}\label{eq:new-lower-bound-c}
		E^\eN_N(\alpha) \ge c(\alpha) N \quad \text{with} \quad
		c(\alpha) := \frac 14 \min\bigl\{ E^\eN_2(\alpha), E^\eN_3(\alpha), E^\eN_4(\alpha) \bigr\}\,.
	\end{equation}
	In particular, 
	\begin{equation}\label{eq:new-lower-bound}
		E^\eN_N(\alpha) \ge c(\alpha) (N-1)_+
	\end{equation}
	for all $N\geq 1$, 
	and therefore by Lemma~\ref{lem:quadratic}
	$$
		E^\eN_N(\alpha) \ge \frac{c(\alpha)} {4} N^2 \bigl(1-O(N^{-1})\bigr)
	$$
	for large $N$. 
	
	From Lemma~\ref{lem:apriori} one can deduce that 
	\begin{equation}\label{eq:calpha}
	c(\alpha) \geq  \frac 14 \min \{ E^\eN_2(\alpha), 0.147\}
	\end{equation}
	where the number $0.147$ is really the positive root of 
	$(\pi+4\sqrt{x})^2 + \frac{9}{4}x = \frac{9}{4}\pi^2$, i.e.
	$$
		x = \pi^2 \frac{ 877 - 96 \sqrt{69} } {5329} \approx 0.147.
	$$
	In combination with Lemma~\ref{lem:n2} and \eqref{eq:E_2-asymptotics}, 
	this concludes the proof.

\subsection{Lieb--Thirring Inequality} 
	Finally, we explain how the above bounds  lead to improvements in the
	local exclusion principle and thus the Lieb--Thirring inequality
	introduced for anyons in \cite{LunSol-13a}.
	Namely, define for any $L^2$-normalized $N$-anyon wave function 
	$\Psi \in \domD_\alpha^N$ and domain $\Omega \subseteq \R^2$
	the local kinetic energy on $\Omega$
	\begin{equation}\label{eq:local-exp-kin-en}
		T_\alpha^\Omega[\Psi] := \sum_{j=1}^N \int_{\R^{dN}} |D_j \Psi|^2 \,\1_\Omega(\bx_j) \,d\bx_1\cdots d\bx_N\,, 
	\end{equation}
	where $\1_\Omega$ denotes the characteristic function of $\Omega$. 
	Applying the bound \eqref{eq:new-lower-bound-c}-\eqref{eq:new-lower-bound}
	as in \cite[Lemma~8]{LunSol-13a} we then obtain the following:
	
\begin{lemma}[Local exclusion principle]
	\label{lem:local-exclusion}
	For any square $Q \subset \R^2$, any $N \ge 1$
	and $L^2$-normalized $\Psi \in \domD_\alpha^N$ with one-particle density
	$\rho_\Psi$ (defined in \eqref{def:rho}),
	we have
	\begin{equation}\label{eq:local-exclusion}
		T_\alpha^Q[\Psi] \ \ge \ 
		\frac{c(\alpha)}{|Q|} \left( \int_Q \varrho_\Psi(\bx) \,d\bx \ - 1 \right)_+,
	\end{equation}
	where $c(\alpha) := \frac 14 \min\bigl\{ E^\eN_2(\alpha), E^\eN_3(\alpha), E^\eN_4(\alpha) \bigr\}$ satisfies \eqref{eq:calpha}.
\end{lemma}

	By applying the method of \cite{LunSol-13a} 
	(see also \cite{Lundholm-17} for a more detailed exposition),
	replacing \cite[Lemma~8]{LunSol-13a} by
	the above bound and using \eqref{eq:calpha} and Lemma~\ref{lem:n2}, 
	one directly obtains the Lieb--Thirring inequality
	of Theorem~\ref{thm:LT} for some universal constant $C>0$.


\end{document}